\documentclass[lettersize,journal]{IEEEtran}

\newtheorem{theorem}{Theorem}

\newtheorem{lemma}{Lemma}   
\newtheorem{definition}{Definition}
\newtheorem{assumption}{Assumption}

\newtheorem{proof}{Proof}
\usepackage[linesnumbered,ruled,boxed]{algorithm2e}
\usepackage{amsmath,amsfonts}
\usepackage{algorithmic}
\usepackage{array}
\usepackage[caption=false,font=normalsize,labelfont=sf,textfont=sf]{subfig}
\usepackage{textcomp}
\usepackage{stfloats}
\usepackage{adjustbox}
 \usepackage{multirow}
 \usepackage{booktabs}       
\usepackage{url}
\usepackage{verbatim}
\usepackage{graphicx}

\def\BibTeX{{\rm B\kern-.05em{\sc i\kern-.025em b}\kern-.08em
    T\kern-.1667em\lower.7ex\hbox{E}\kern-.125emX}}
\usepackage{balance}
\begin{document}
\title{Information-Theoretic Generalization Analysis for Topology-aware Heterogeneous Federated  Edge Learning over Noisy Channels}
\author{Zheshun Wu, Zenglin Xu, \IEEEmembership{Senior Member, IEEE}, Hongfang Yu, \IEEEmembership{Member, IEEE}, and Jie Liu, \IEEEmembership{Fellow, IEEE}
\thanks{This work was partially supported by a key program of fundamental research from Shenzhen Science and Technology Innovation Commission (No. JCYJ20200109113403826), the Major Key Project of PCL (No. 2022ZD0115301), and an Open Research Project of Zhejiang Lab (NO.2022RC0AB04). (\emph{Corresponding author: Zenglin Xu.})}
\thanks{Zheshun Wu, Zenglin Xu, and Jie Liu are with the School of Computer Science and Technology, Harbin Institute of Technology Shenzhen, Shenzhen 518055, China (e-mail:
wuzhsh23@gmail.com; xuzenglin@hit.edu.cn; jieliu@hit.edu.cn).}
\thanks{Hongfang Yu is with the School of Information and Communication Engineering, University of Electronic Science and Technology of China, Chengdu
610054, China (e-mail: yuhf@uestc.edu.cn).}}

\markboth{Journal of \LaTeX\ Class Files,~Vol.~18, No.~9, September~2020}%
{How to Use the IEEEtran \LaTeX \ Templates}

\maketitle

\begin{abstract}
With the rapid growth of edge intelligence, the deployment of federated learning (FL) over wireless networks has garnered increasing attention, which is called Federated Edge Learning (FEEL). In FEEL, both mobile devices transmitting model parameters over noisy channels and collecting data in diverse environments pose challenges to the generalization of trained models. Moreover, devices can engage in decentralized FL via Device-to-Device communication while the communication topology of connected devices also impacts the generalization of models. Most recent theoretical studies overlook the incorporation of all these effects into FEEL when developing generalization analyses and ignore designing algorithms to enhance the generalization of models based on their analysis.  In contrast, our work presents an information-theoretic generalization analysis for topology-aware FEEL with data heterogeneity and noisy channels. Additionally, we propose a novel regularization method called Federated Global Mutual Information Reduction (FedGMIR) to enhance the performance of models based on our analysis. Numerical results validate our theoretical findings and provide evidence for the effectiveness of the proposed method.
\end{abstract}


\section{Introduction}
\IEEEPARstart{D}{ue} to the development of edge intelligence in the Beyond fifth-generation (B5G) mobile communication, federated learning (FL), a privacy-preserved  distributed machine learning framework, has been considered as a key technique to assist intelligent tasks over wireless networks~\cite{DBLP:journals/twc/AmiriG20}. Mobile devices utilize local datasets for collaborative  training, a practice known as Federated Edge Learning (FEEL). 
FEEL can be deployed in one of two ways: devices can either upload models to an edge server through Centralized Federated Learning (CFL) or they can exchange models directly with other devices using Device-to-Device (D2D) communication in Decentralized Federated Learning (DFL)~\cite{DBLP:journals/jsac/XingSB21}. A key challenge in FL is ensuring that models trained on data collected by clients can generalize effectively to unseen test data. This challenge has led to in-depth theoretical analyses on the generalization capabilities of FL~\cite{DBLP:conf/iclr/HuLL23,DBLP:journals/corr/abs-2306-03824,wu2023federated}.


A key obstacle to FEEL's generalization lies in the requirement for devices to relay models via wireless channels. Disruptions from wireless channels like channel noise and path loss not only affect the accurate reception of models but also have downstream consequences on model aggregation and the overall performance of FL~\cite{DBLP:journals/jsac/XingSB21}. Besides, communication topologies affect the speed of information exchange among D2D devices, thereby affecting the transmission of models and the performance of  DFL~\cite{sun2021stability}. Additionally, data distributions across devices are diverse since data collection in distinct environments, leading to the Non-Independent and Identically Distributed (Non-IID) phenomenon~\cite{DBLP:conf/icde/LiDCH22}. Consequently, it is vital to investigate the impacts of the aforementioned factors on the generalization of FEEL through theoretical analysis, as well as to design algorithms that enhance the  performance of models for better deploying FL systems in wireless networks.

Recently, several works attempt to develop the generalization analysis for FL. For instance, \cite{DBLP:conf/iclr/HuLL23} provides an algorithm-independent generalization bound based on the restrictive meta-distribution assumption. \cite{DBLP:journals/corr/abs-2306-03824} utilizes the algorithmic stability to analyze the generalization but they only consider FL deployment over noise-free links. \cite{DBLP:journals/entropy/BarnesDP22,DBLP:conf/spawc/YagliDP20} analyze the generalization of FL via information-theoretic generalization analysis while lacking research on the effect of data heterogeneity and ignoring the DFL scenario. In brief, these studies do not consider the generalization of realistic FEEL scenarios by integrating key factors in wireless networks. Moreover, they neglect how to devise methods for improving the  performance of FL by reducing their proposed generalization bounds.

Given the combined impact of data heterogeneity, noisy channels, and communication topologies on the generalization of models, this paper develops information-theoretic generalization bounds for both CFL and DFL scenarios of FEEL and reveals the trade-offs between the generalization of DFL and the communication cost of D2D networks. To tackle these effects, we present a  method called Federated Global Mutual Information Reduction (FedGMIR) to improve the generalization performance of models. Numerical results validate our theoretical results and verify the effectiveness of our methods.

\section{Theoretical Framework}
\subsection{Problem Formulation}

Given an instance space $\mathcal{Z}$ and a hypothesis space $\mathcal{W}$,  a non-negative loss function is defined as $\ell:\mathcal{W} \times \mathcal{Z} \rightarrow \mathbb{R}$. Random variables $Z$ and $W$ have support over spaces $\mathcal{Z}$ and $\mathcal{W}$ respectively. Assume there are $N$ devices participating in FL. We denote the local distribution of device $i$ as $\mu_i$ and  the corresponding local training set $S_i=\{Z^i_1,Z^i_2,...,Z^i_n\}$  is  drawn from $\mathcal{Z}$ based on $\mu_i$ independently, i.e., $Z^i_j\stackrel{i.i.d.}{\sim}\mu_i$. The size of $S_i, \forall i \in [N]$ is $n$. We assume that the local data sources are independent across devices while the distributions are different, i.e., $\mu_i\neq \mu_j,\forall i,j \in [N], i \neq j$, meaning that devices collect data  in different manners independently.

The expected risk minimization is a crucial objective in machine learning~\cite{DBLP:journals/entropy/BarnesDP22,DBLP:conf/isit/PensiaJL18,DBLP:conf/colt/Neu21} and our focus is  to find a model $W \in \mathcal{W}$ minimizing the global expected risk in FL. The global expected risk in FL is defined as $\frac{1}{N}\sum_{i=1}^N \mathcal{L}_{\mu_i}(W)=\frac{1}{N}\sum_{i=1}^N\mathbb{E}_{Z\sim\mu_i}\ell(W,Z)$.  In practice, minimizing the global expected risk is achieved  approximately  by optimizing  the tractable global empirical risk,  defined as $\frac{1}{N}\sum_{i=1}^N \mathcal{L}_{S_i}(W)=\frac{1}{N}\sum_{i=1}^N \frac{1}{n}\sum_{j=1}^n\ell(Z_j^i,W)$, where $Z_j^i$ denotes $j$-th sample  in $i$-th dataset. We then define the approximate error that arises when training the  model by drawing samples from empirical distributions $\{S_i\}_{i=1}^N$ instead of  true distributions $\{\mu_i\}_{i=1}^N$, which is called the generalization error in FL:
\begin{definition}[Generalization error in FL]
\begin{equation}
    gen(\mu,P_{W|S}):=\mathbb{E}_{S,W}\Big[\frac{1}{N}\sum_{i=1}^N \mathcal{L}_{\mu_i}(W)-\frac{1}{N}\sum_{i=1}^N\mathcal{L}_{S_i}(W)\Big],
\end{equation}
where $\mu=(\mu_1,\mu_2,...,\mu_N)$ and $S=(S_1,S_2,...,S_N)$.
\end{definition}

    Suppose $\ell(Z_i,w)$ is $R$-sub-Gaussian for any $w\in \mathcal{W}$ with respect to $Z_i,\forall i \in [N]$, i.e., $\mathbb{E}[\exp(\lambda(\ell(w,Z_i)-\mathbb{E}\ell(w,Z_i)))]\leq \exp(\frac{\lambda^2 R^2}{2})$, based on the results of~\cite{DBLP:journals/entropy/BarnesDP22}, we can immediately have the below lemma:
    \begin{lemma}[Information-theoretic generalization bound for federated learning]\label{lemma:GFL}
    \begin{equation}\label{eq:lemm1}
  |gen(\mu,P_{W|S})|   \leq \sqrt{\frac{2R^2}{nN}\sum_{i=1}^NI(W;S_i)},
    \end{equation}
    where $I(W;S_i)$ represents the mutual information between the local training dataset $S_i$  and the learned global model $W$. This result implies that reducing the dependency between the global model and datasets $\{S_i\}_{i=1}^N$ can avoid  overfitting and improve the generalization of FL, aligning with our intuition.
\end{lemma}

In the following, we develop the generalization analysis of models trained in CFL and DFL based on  over-the-air computation techniques~\cite{DBLP:journals/jsac/DuJJSH23} over noisy channels. In this paper, we assume that devices update  models via stochastic gradient descent (SGD) and  they only transmit models over Additive White Gaussian Noise (AWGN) channels.  The update rules of  the two scenarios will be introduced below.

\subsubsection{Centralized Federated Learning}
We denote the global model at $t$-th round as $W_t$ and $W_0$ is set as the $\mathbf{0}$~\cite{sun2021stability}. At each communication round, each device samples a data point $Z_t^i \in S_i$ (or minibatch) and computes the local gradient $\nabla\ell(W_{t},Z^i_t)$ based on the received global model $W_t$. Since a greater transmit power at the edge server than at devices,  downlink channels  are assumed to be noise-free in CFL~\cite{DBLP:journals/twc/AmiriG20}. Then all the devices upload  updated  models $W^i_{t}=W_t-\eta_t\nabla\ell(W_{t},Z^i_t)$ to the server via the AWGN channel. The server aggregates received models  to update the global model following:
\begin{equation}\label{CFLrule}
W_{t+1}=\frac{1}{N}\sum_{i=1}^NW^i_{t}+\epsilon_t=W_t-\frac{\eta_t}{N}\sum_{i=1}^N \nabla\ell(W_{t},Z^i_t)+\epsilon_t,
\end{equation}
where $\eta_t$ is the learning rate at $t$-th round. $\epsilon_t \sim \mathcal{N}(0,\sigma^2I_d)$ denotes the isotropic Gaussian noise. $I_d$ is an identity matrix and $d$ is the dimension of parameters vector and noise.

\subsubsection{Decentralized Federated Learning}
Then we move to introduce the details of DFL considered in this paper. We assume the decentralized optimization is  associated with a matrix $\boldsymbol{\Theta}=[\theta_{i,j}]\in\mathbb{R}^{N\times N}$ , which
is constructed by  devices rooted on a specific graph structure.  More specifically, we consider the graph $\mathcal{G}=(\mathcal{V},\mathcal{E})$ with vertex set $\mathcal{V}=\{1,2,...,N\}$ and edge set $\mathcal{E}\subset \mathcal{V} \times \mathcal{V} $ with edge $(i,j)\in\mathcal{E} $ representing the communication link between device $i$ and $j $. 

We further consider $\boldsymbol{\Theta}$ is a doubly stochastic matrix, satisfying: a) If $i\neq j $ and $(i,j) \notin \mathcal{E}$, then $\theta_{i,j}=0$; otherwise $\theta_{i,j}> 0$; b) $\boldsymbol{\Theta}=\boldsymbol{\Theta}^T, \boldsymbol{\Theta}\mathbf{1}_N=\mathbf{1}_N$, and ${\mathbf{1}_N}^T\boldsymbol{\Theta}={\mathbf{1}_N}^T$, where ${\mathbf{1}_N} \in \mathbb{R}^N$ is the vector of all ones. A significant value that characterizes $\boldsymbol{\Theta}$ is $\lambda:=\max\{|
\lambda_2(\boldsymbol{\Theta})|,|\lambda_N(\boldsymbol{\Theta})|\}$, where $|\lambda_i(\boldsymbol{\Theta})|$ denotes $i$-th largest eigenvalue of $\boldsymbol{\Theta}$.  $ 0\leq \lambda <1$ since the definition of the doubly stochastic matrix and $\lambda=0$ indicates that the communication topology is fully-connected.

Referring to~\cite{sun2021stability}, the local training procedure of devices in DFL follows: 1) At the $t$-th round, device $i$ leverages the local model $W_t^i$ to compute the local gradient $\nabla\ell(W_t^i,Z_t^i)$ similar to CFL. 2) device  $i$ replaces its local models with the weighted average of $\sum_{j=1}^N\theta_{i,j}W^j_{t}$ over AWGN channels. 3) Then device $i$ updates its model $W_{t+1}^i$ following:
\begin{equation}\label{DFLrule}
W_{t+1}^i=\sum_{j=1}^N\theta_{i,j}W^j_{t}-\eta_t\nabla\ell(W^i_{t},Z^i_t)+\epsilon_t^i,
\end{equation}
where $\theta_{i,j}$ is the element of  the doubly stochastic matrix $\boldsymbol{\Theta}$ and $\epsilon_t^i \sim \mathcal{N}(0,\sigma_i^2I_d)$ denotes the isotropic Gaussian noise.  

In this paper, we only consider the generalization performance of the average model $W_t=\frac{1}{N}\sum_{i=1}^N W_{t}^i$ in DFL. Besides, we consider four specific graph structures used as the communication topologies and they are shown in Fig.~\ref{fig:topology}.
\begin{figure}[t]
\captionsetup[subfloat]{font=tiny}	
\centering
	\subfloat[Complete]{\includegraphics[width = 0.08\textwidth]{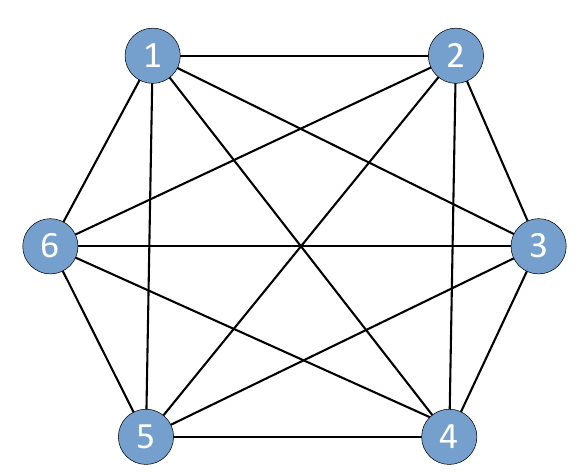}}
	\hspace{0.6cm}\subfloat[Ring]{\includegraphics[width = 0.08\textwidth]{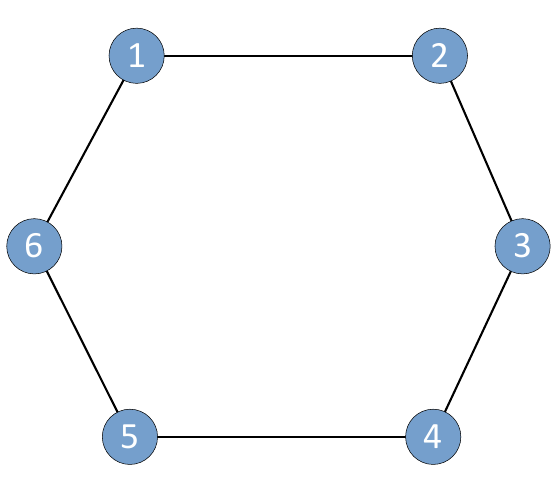}}
 	\hspace{0.6cm}\subfloat[Star]
 {\includegraphics[width = 0.08\textwidth]{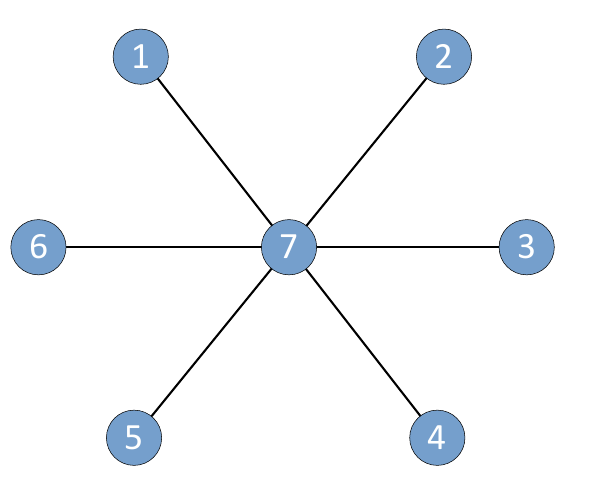}}
	\hspace{0.6cm} \subfloat[Random]{\includegraphics[width = 0.08\textwidth]{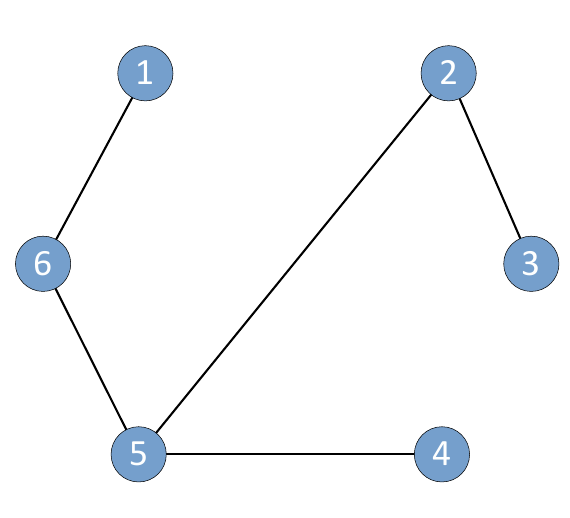}}
\caption{Communication topology for DFL considered in this paper.}
\label{fig:topology}
\vspace{-0.5cm}
\end{figure}

\subsection{Generalization Analysis}
This section proposes the generalization  bounds for both CFL and DFL introduced in the previous section. The generalization analysis utilizes the below  assumptions widely used in the theoretical studies of FL.

\begin{assumption}[$L$-Lipschitz]\label{assumLip}
$\ell(w,z)$ is differentiable with respect to $w    $ and $L$-Lipschitz for every $z$, i.e. 
, for any $ w_1,w_2\in \mathcal{W}$, $|\ell(w_1,z)-\ell(w_2,z)|\leq L\Vert w_1-w_2 \Vert$, for a constant $L \geq 0$.
\end{assumption}
\begin{assumption}[Bounded variance]\label{assumpbounded}
For any $w\in \mathcal{W}$, the variance of the stochastic gradient is bounded by $\xi_i$, i.e., $\mathbb{E}_{Z_i \sim \mu_i}\Vert\nabla\ell(Z_i,w) -\nabla\mathcal{L}_{\mu_i}(w)\Vert \leq \xi_i$,
where $\mathcal{L}_{\mu_i}(w)=\mathbb{E}_{Z_i \sim \mu_i}[\ell(Z_i,w)]$ denotes the expected risk of $\mu_i$.
\end{assumption}

These assumptions are widely used in the theoretical analysis of FL~\cite{DBLP:conf/iclr/HuLL23, DBLP:journals/corr/abs-2306-03824,DBLP:journals/entropy/BarnesDP22}. Under these assumptions, we  present the information-theoretic generalization bound for both CFL and DFL over AWGN channels below.  The proof sketch is provided alongside our theoretical results  and the detailed theoretical proofs are included in the appendix.  

\begin{theorem}[Generalization bound for CFL over AWGN channels.] \label{thm1} Let $W$ be the final iterate $W_T$ of sequence  $W^{(T)}=(W_1,W_2,...,W_T)$ generated by CFL over AWGN channels based on Eq.~\eqref{CFLrule}. The generalization bound of such $W$ is bounded as,
\begin{equation}\label{eq:thm1}
    \begin{aligned}
          |gen_{CFL}(\mu,P_{W|S})|^2   
        &\leq \sum_{t=1}^T \frac{R^2\eta_t^2}{\sigma^2N^3n}\sum_{i=1}^N\Big[\xi_i^2+L^2(4D_i^2 +1)\Big],
    \end{aligned}
\end{equation}
where  $D_i=\frac{1}{2}\int_{z \in \mathcal{Z}}|d\mu_i(z)-\frac{1}{N}\sum_{i=1}^Nd\mu_i(z)|$ is the total variation of distributions $\mu_i$ and $\frac{1}{N}\sum_{i=1}^N\mu_i$, measuring the degree of data heterogeneity.
\end{theorem}

 \emph{Proof sketch and Remark for  Theorem~\ref{thm1}:}
We utilize the independence of $\{S_i\}_{i=1}^N$, the data processing inequality $I(W_t;S_i)\leq I(W_t;Z_t^i)$ and the Markov structure of $\{W_t\}_{t=1}^T$ to upper-bound the quantity $\sum_{i=1}^NI(W;S_i)$ in Lemma~\ref{lemma:GFL} as $\sum_{t=1}^TI(W_{t+1};(Z_t^i)_{i=1}^N|W_t)$. By employing the update rule of CFL and the property of AWGN channels, we derive the result in Eq.~\eqref{eq:thm1}. Theorem~\ref{thm1} demonstrates that the data heterogeneity  $D_i$ deteriorates the generalization. Besides, a certain degree of channel noise $\epsilon_t$ with variance $\sigma^2$    is beneficial for generalization as it introduces randomness during training, thereby avoiding overfitting. However,  excessive noise degenerates the training convergence and affects the performance~\cite{DBLP:journals/jsac/XingSB21}.

Next, we proceed to introduce Theorem~\ref{thm2}, which presents the generalization bound of the average model $W_t=\frac{1}{N}\sum_{i=1}^NW_t^i$ trained by  DFL over AWGN channels.
\begin{theorem}[Generalization bound for DFL over AWGN channels.] \label{thm2} Let $W$ be the final average model $W_T=\frac{1}{N}\sum_{i=1}^N W_T^i$ of sequence  $W^{(T)}=(W_1,W_2,...,W_T)$ generated by  DFL over AWGN channels based on Eq.~\eqref{DFLrule}. The generalization bound of this $W_T$ is bounded as,
\begin{equation}\label{eq:thm2}
\begin{aligned}
         |gen_{DFL}(\mu,P_{W|S})|^2 &\leq \sum_{t=1}^T \Big[\frac{2R^2\lambda_{T}\mathcal{V}\sum_{k=1}^{t}\eta_{t-k}^2}{\sigma^2N^3n}\\
     &\quad +\frac{2R^2\lambda_{T}d}{N^3n}+\frac{2R^2\mathcal{V}\eta_t^2}{\sigma^2Nn}\Big],
\end{aligned}
\end{equation}
where   $\lambda_T$ denotes $\max_{t\in[T]}\sum_{k=1}^t\lambda^{2k}$, $\sigma^2=\sum_{i=1}^N\sigma_i^2$ denotes the total effect from channel noises and $\mathcal{V}$ represents the term $\sum_{i=1}^N[\xi_i^2+L^2(4D_i^2+1)]$, measuring the total effect  from  the variance of stochastic gradients and the data heterogeneity.
\end{theorem}

 \emph{Proof sketch and Remark for Theorem~\ref{thm2}:}
Following~\cite{DBLP:conf/colt/Neu21}, we convert $I(W_{t+1};\{Z_t^i\}_{i=1}^N|W^{(t)})$ into the relative entropy between the conditional distribution $P_{W_{t+1}|W^{(t)},(Z_t^i)_{i=1}^N}$ of $W_{t+1}$  and $P_{\Tilde{W}_{t+1}|W^{(t)}}$ of $\Tilde{W}_{t+1}$, a ``ghost" iterate  trained by virtual datasets $\{\Tilde{S}_i\}_{i=1}^N$ independently sampled from $\{\mu_i\}_{i=1}^N$. We further leverage the update rule of DFL and the properties of mixing matrices to derive the result in Eq.~\eqref{eq:thm2}. Theorem~\ref{thm2} implies that though decentralization decreases communication costs, it impairs the generalization since the sparsity of topologies leads to a larger $\lambda$. Eq.~\eqref{eq:thm2} also shows that the model scale $d$ exacerbates this deterioration, implying that exchanging more packets via D2D links worsens the generalization.

The trade-offs between the generalization error of DFL and the communication budget is crucial because of the limited bandwidth in D2D networks. Our focus thus lies in the interplay between the number of communication round $T$ and the generalization error $|gen_{DFL}(\mu,P_{W|S})|$ in Eq.~\eqref{eq:thm2} denoted as $\varepsilon$ below. If we set the learning rate $\eta_t$ as a constant $\eta$,  $\varepsilon$ then satisfies $ \varepsilon=\mathcal{O}\big(R\sqrt{\frac{\eta^2\mathcal{V}\lambda_T}{\sigma^2N^3n}T^2+(\frac{\eta^2\mathcal{V}}{\sigma^2Nn}+\frac{\lambda_Td}{N^3n})T}\big)$, indicating that $\varepsilon$ is accumulated  as $T$ increases. Moreover, both larger model scale $d$ involving more data packets to be exchanged, and sparser topologies with larger $\lambda_T$ lead to a rapid growth of $\varepsilon$. Additionally, setting $\eta_t=\frac{\eta}{t}$ yields $ \varepsilon=\mathcal{O}\big(R\sqrt{(\frac{\eta^2\mathcal{V}\lambda_T}{\sigma^2N^3n}+\frac{\lambda_Td}{N^3n})T+\frac{\eta^2\mathcal{V}}{\sigma^2Nn}}\big)$, signifying that  $\varepsilon$ increases more gradually with  decreasing  $\eta_t$.

\section{Method}
In this section, we devise a novel federated optimization method to reduce the information-theoretic generalization bound in Lemma~\ref{lemma:GFL} for improving the  performance of FL over noisy channels. Recently,  a PAC-Bayes information bottleneck framework  was introduced in~\cite{DBLP:conf/iclr/0008HKSC022} to approximately reduce $I(W;S)$  in the information-theoretic generalization bound via optimizing the loss function with a regularization term $[\nabla\ell(W,Z_i)^T(W-W_0)]^2, Z_i \in S$ by stochastic gradient langevin dynamics (SGLD), where $W_0$ is a prior model. Motivated by this approach, we extend this concept and  propose a federated optimization method called Federated Global Mutual Information Reduction (FedGMIR) to reduce the generalization error of  FL over noisy channels via approximately decreasing the global mutual information $\frac{1}{N}\sum_{i=1}^N I(S_i;W)$ included in the generalization bound.

Now we introduce the details of the  proposed method formally: Each device modifies its local optimization objectives as:  $\ell(Z^i_t,W)+\beta \big[\nabla\ell(Z^i_t,W)^T(W-W_0)\big]^2$, where $\beta$ is a balanced factor and $Z^i_t$ represents the minibatch of $i$-th local dataset in $t$-th round.  $W_0$ is set as the initial model. Based on the modified local objectives, the update rules of the considered models of CFL and DFL are thus converted as:
\begin{equation}\label{IB}
\begin{aligned}
W_{t+1}&=W_t-\frac{\eta_t}{N}\sum_{i=1}^N \nabla_{W_{t}}\big[\ell(Z^i_t,W_t)+\beta R(W_t,Z_t^i)\big]+\epsilon_t\\
 W_{t+1}&=\frac{1}{N}\sum_{i=1}^N\Big\{\sum_{j=1}^N\theta_{i,j}W^j_{t}-\eta_t \nabla_{W^i_{t}}\big[\ell(W^i_{t},Z^i_t)\\
 &\quad+\beta  R(W_t^i,Z_t^i)\big]+\epsilon_t^i\Big\},
\end{aligned}
\end{equation}
where $R(W_t,Z_t^i)=\big[\nabla\ell(Z^i_t,W_t)^T(W_t-W_0)\big]^2$ and $R(W_t^i,Z_t^i)=\big[\nabla\ell(Z^i_t,W_t^i)^T(W_t^i-W_0)\big]^2$. The channel noise $\epsilon_t$ and $\epsilon_t^i$ can help to attain the optimal  posterior inference $P_{W|S}$ similar to the SGLD algorithm used in~\cite{DBLP:conf/iclr/0008HKSC022}. Intuitively, to minimize  regularization terms in Eq.~\eqref{IB} implies making learned models  orthogonal to gradients for reducing the dependency between models and training samples, since calculated gradients incorporate information from local datasets.  FedGMIR can thus enhance the generalization  of models via  implicitly reducing the global mutual information $\frac{1}{N}\sum_{i=1}^N I(S_i;W)$ in the generalization bound in Lemma~\ref{lemma:GFL}.

\section{Numerical Results}   
In this section, we conduct numerical experiments to evaluate the  performance of the proposed FedGMIR  with  typical FL methods on two widely used image classification datasets in FL, namely FMNIST and CIFAR-10 respectively~\cite{DBLP:conf/icde/LiDCH22}. Additionally, we numerically verify our theoretical results of the generalization of FEEL on CIFAR-10.

The experimental setup is described below.  The number of devices is set to $50$.    Multi-layer perceptron (MLP) and Convolutional Neural Network (CNN) are selected  as the backbone models for FMNIST and CIFAR-10 respectively.   Devices perform SGD to update local models with a mini-batch size of $128$ and we tune the learning rate over $\{0.001,0.01,0.1\}$.   For modeling the data heterogeneity, we distribute the datasets to devices  using the Dirichlet distribution division widely used in Non-IID FL, where a smaller parameter $\alpha$ of Dirichlet distribution means a higher degree of heterogeneity~\cite{DBLP:conf/icde/LiDCH22}. Four graph structures illustrated in Fig.~\ref{fig:topology} are selected as the topologies in  DFL  for investigating the effects of the structure of topology. We employ the Erdős-Rényi random graph to construct the random graph~\cite{DBLP:conf/nips/LeiYHCS20}, with an edge creation probability set to  $0.01$. We test the global model or the average model on the test set following the same distribution as that of all the devices.  The corresponding test accuracy achieved by the model is considered as the measure of generalization performance.

    \begin{figure*}[ht]
\captionsetup[subfloat]{font=scriptsize}	
\centering
\vspace{-0.6cm}
\subfloat[ Data heterogeneity in CFL]{\includegraphics[width = 0.2\textwidth]{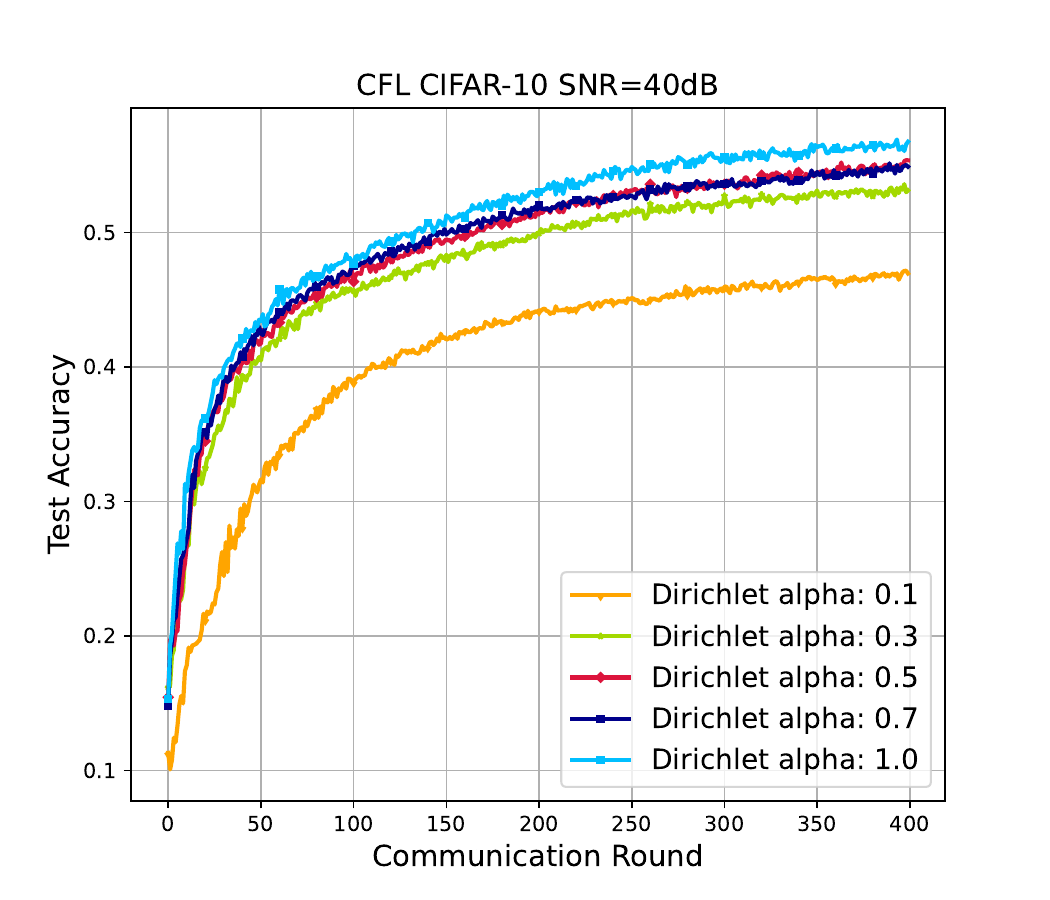}}
\hspace{0.4cm}
   	\subfloat[ Data heterogeneity in DFL]{\includegraphics[width = 0.2\textwidth]{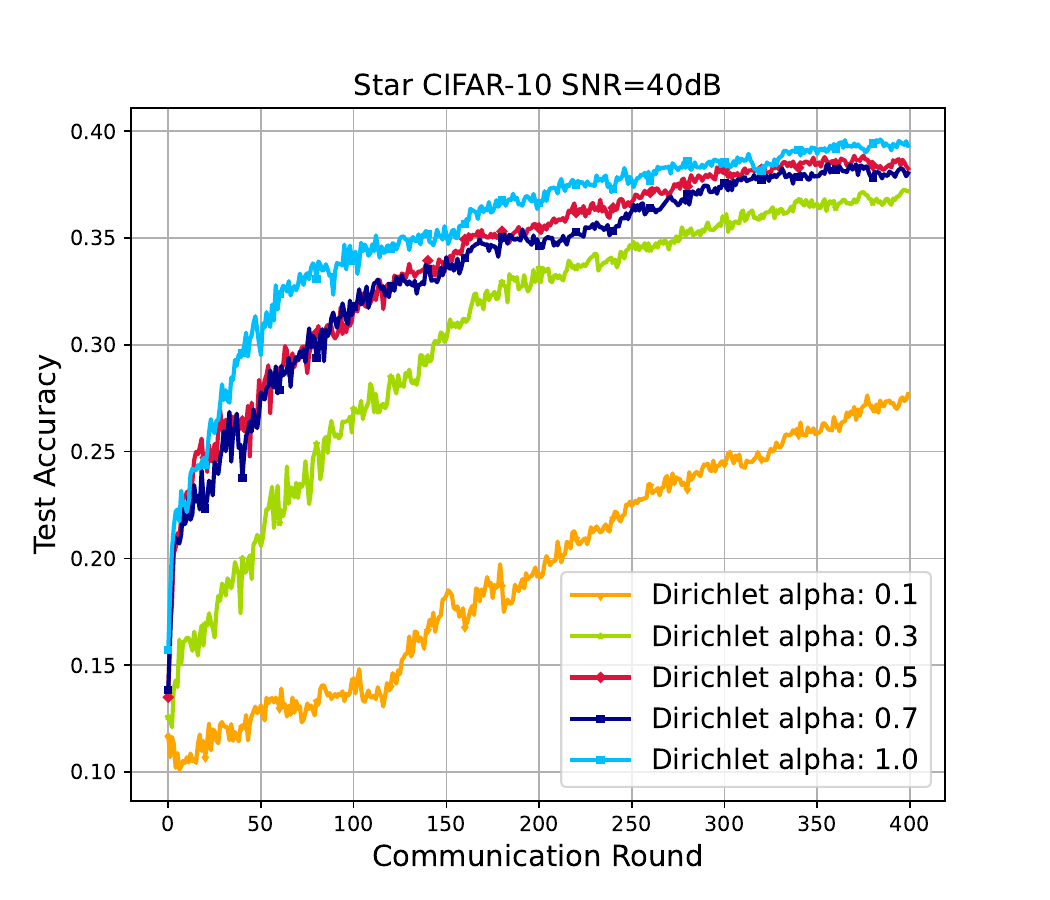}}
    \hspace{0.4cm}
	\subfloat[Communication topology]{\includegraphics[width = 0.2\textwidth]{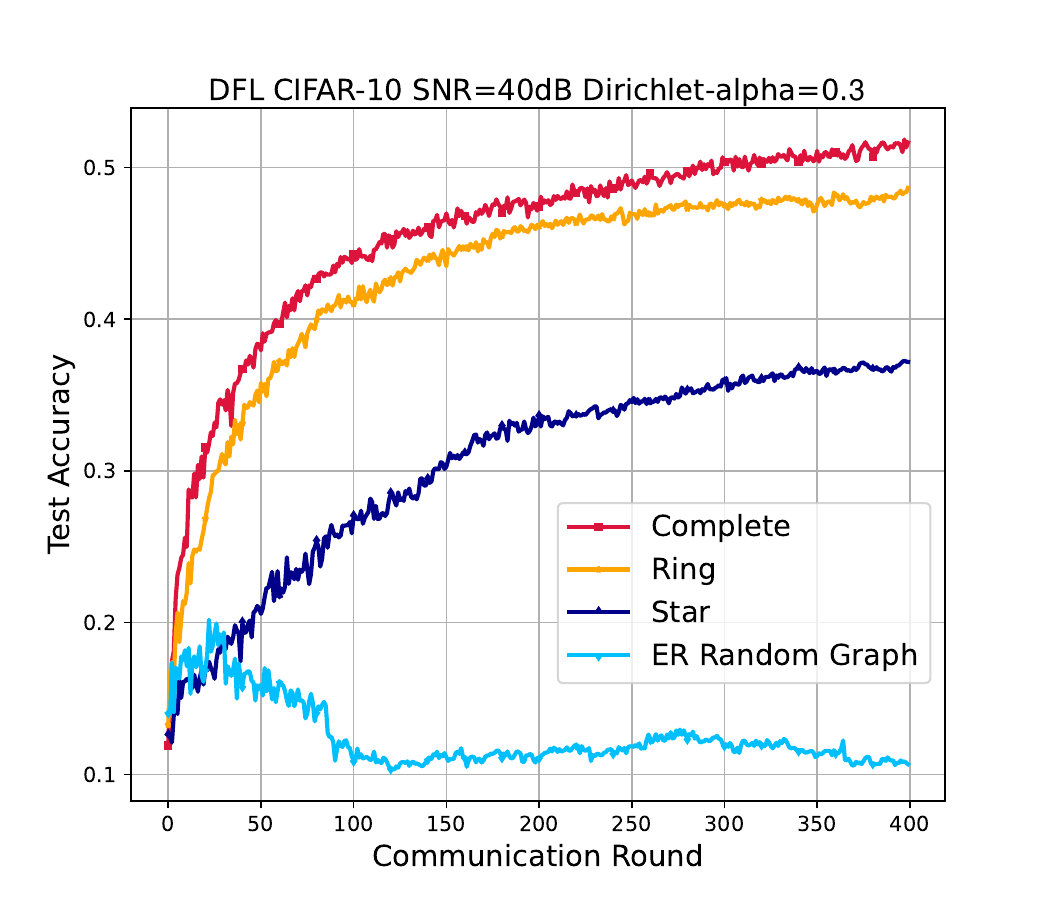}}
 \hspace{0.4cm}
 \subfloat[Signal-to-noise ratio]{\includegraphics[width = 0.2\textwidth]{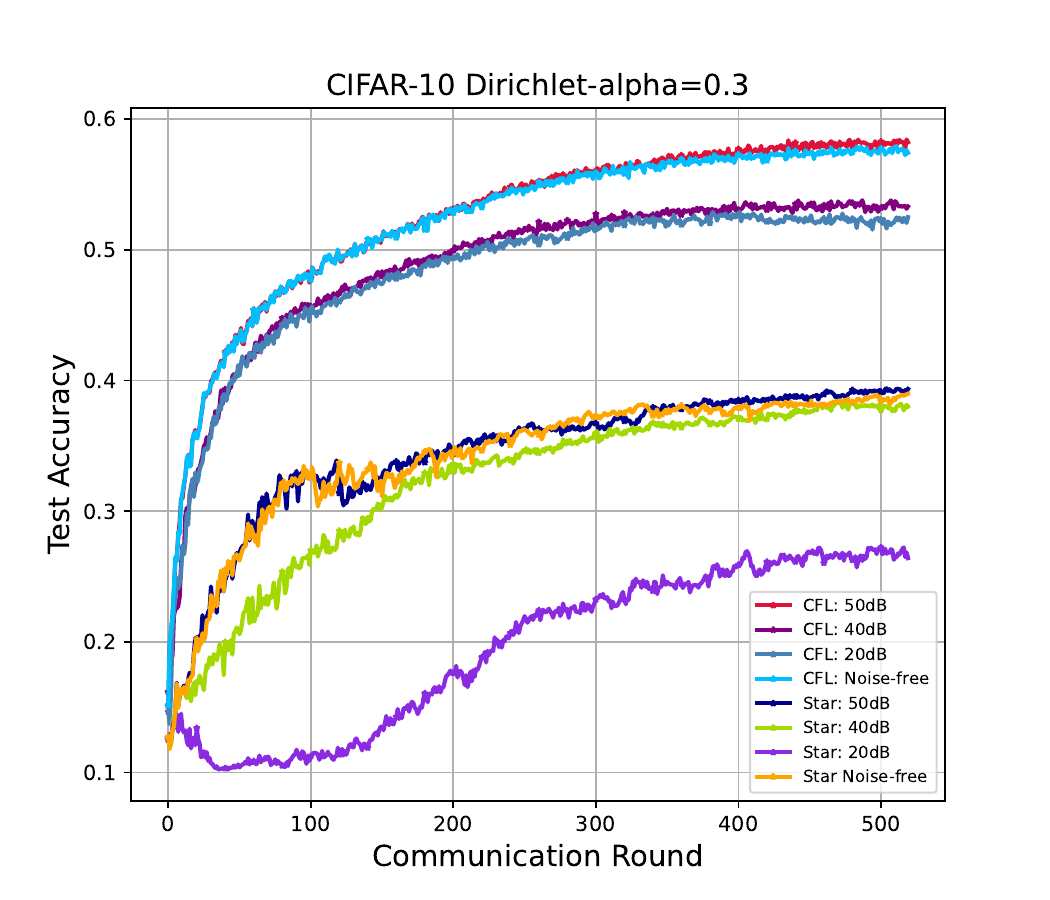}}
\caption{Test accuracy of different effects on the generalization of FL for CIFAR-10.}
\vspace{-0.4cm}
\label{fig:effect}
\end{figure*}

Fig.~\ref{fig:effect} illustrates the convergence  of test accuracy on CIFAR-10 considering various effects included in the proposed generalization bound. We first examine the impact of data heterogeneity   in CFL and DFL scenarios, while keeping the signal-to-noise ratio (SNR) at $40$dB and selecting the star topology as the  topology in DFL. Consistent with our theoretical findings rooted on the proposed generalization bound,  increasing the data heterogeneity $D_i$  negatively affects generalization, as demonstrated in Fig~\ref{fig:effect} (a) and~\ref{fig:effect} (b). Next, we focus on the effect of different topologies on generalization. As mentioned earlier, we employ four structures depicted in Fig.~\ref{fig:topology} as the topologies  in DFL and the values of $\lambda$ in Theorem~\ref{thm2}  indicate the sparsity of the corresponding matrices $\boldsymbol{\Theta}$, where $0=\lambda_{complete} <  \lambda_{ring} \leq \lambda_{star} <\lambda_{random} < 1$ in our settings. According to Theorem~\ref{thm2}, a smaller $\lambda$, implying a more densely connected  topology among devices, can reduce the generalization bound. Consequently, though decentralization can decrease the bandwidth cost over transmission channels, it adversely affects the generalization of DFL, as illustrated in Fig.~\ref{fig:effect} (c). Fig.~\ref{fig:effect} (d) demonstrates  the effect of SNR (noise scale $\sigma$) on the generalization. Based on our theoretical analysis, the introduced randomness of channel noise can improve the generalization to some extent. In Fig.~\ref{fig:effect} (d), the test accuracy of models under SNR$=55$dB  is slightly higher compared with performing FL without channel noise. Nonetheless, lower SNR such as $40$dB or $20$dB severely hinder the training process of FL, resulting in lower test accuracy.
\begin{figure}[ht]
\captionsetup[subfloat]{font=scriptsize}	
\vspace{-0.4cm}
\centering
	\subfloat[ CIFAR-10 CFL]{\includegraphics[width = 0.18\textwidth]{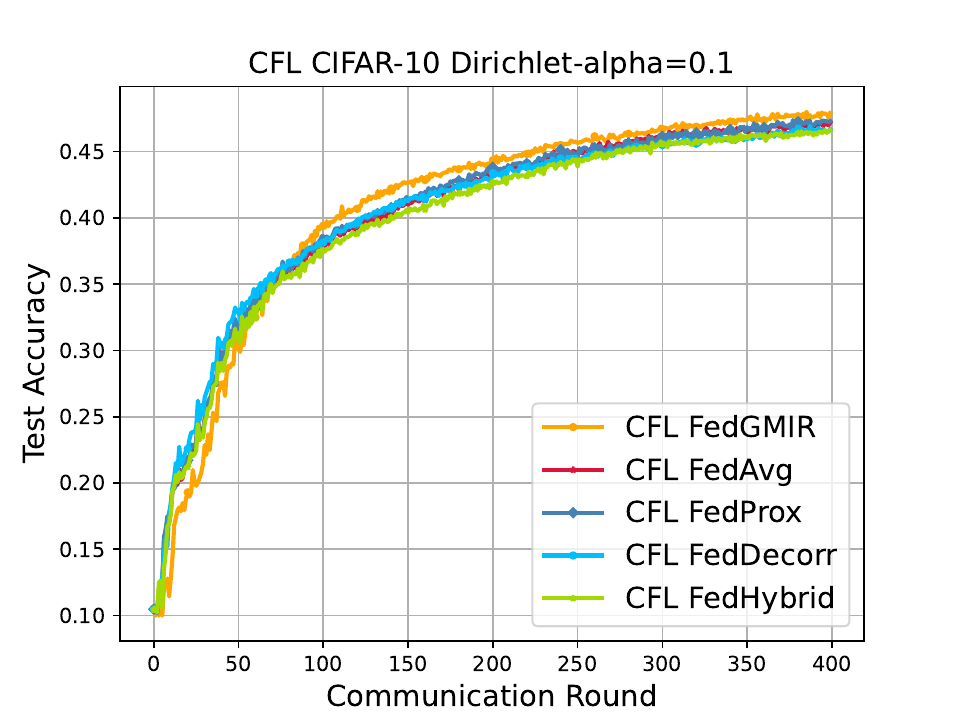}}
 \hspace{0.3cm}
 \subfloat[CIFAR-10 DFL]{\includegraphics[width = 0.18\textwidth]
 {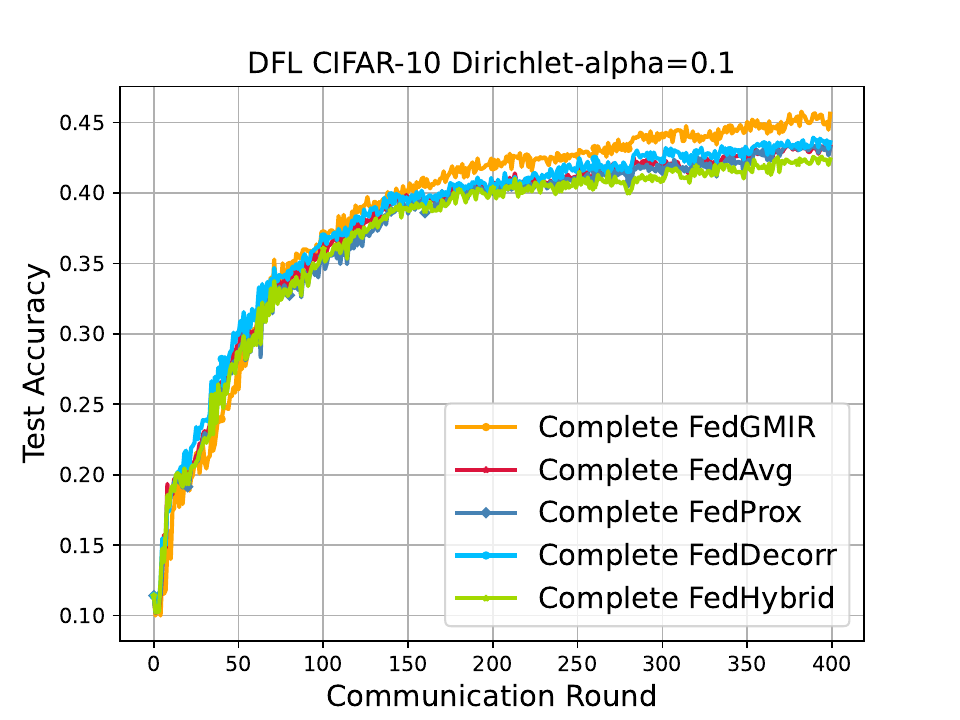}}
   \vspace{-0.05cm}
  \subfloat[FMNIST CFL]{\includegraphics[width = 0.18\textwidth]
 {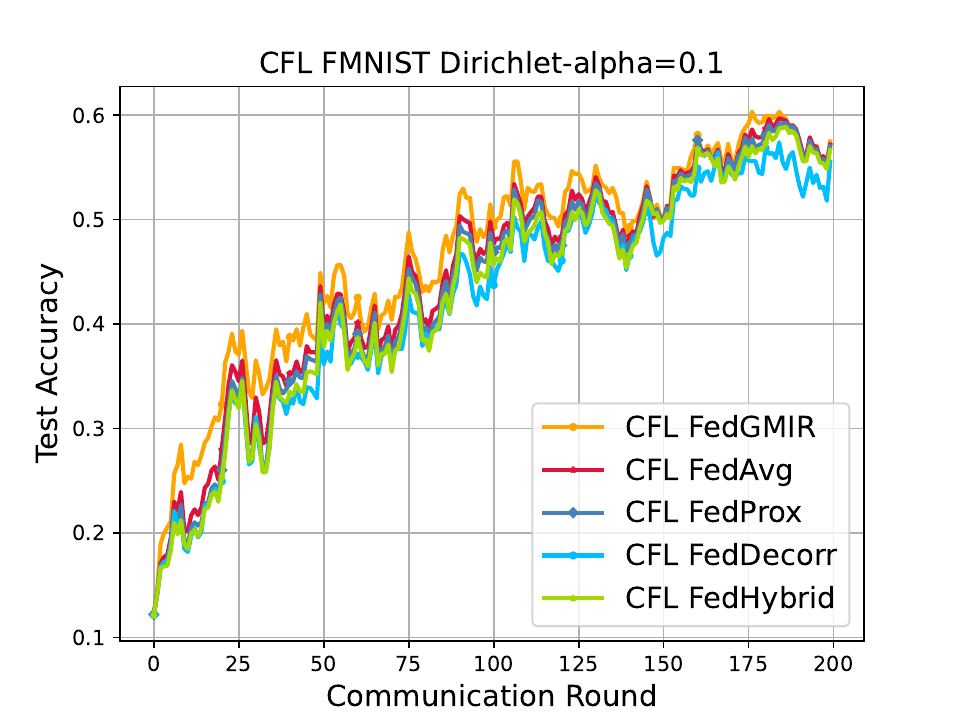}}
  \hspace{0.3cm}
  \subfloat[FMNIST DFL]{\includegraphics[width = 0.18\textwidth]
 {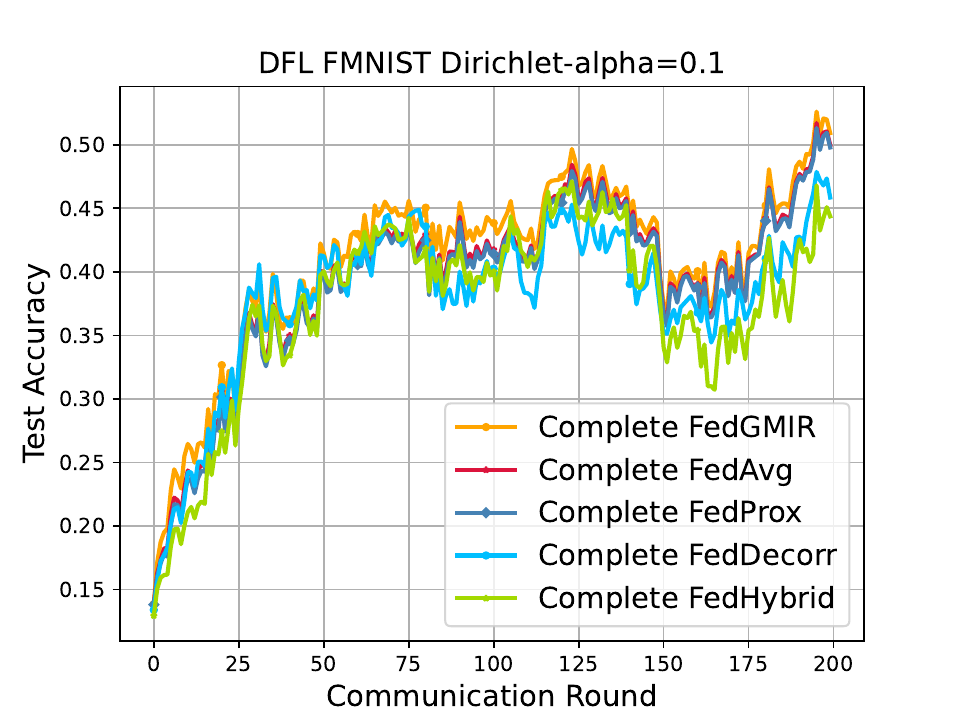}}
\caption{Test accuracy of different federated optimization methods.}
\label{fig:method}
\vspace{-0.2cm}
\end{figure}

We then conduct experiments to validate the effectiveness of FedGMIR on CIFAR-10 and FMNIST, presenting the results in Fig.~\ref{fig:method}. We select FedAvg~\cite{DBLP:conf/aistats/McMahanMRHA17}, FedProx~\cite{li2020federated}, FedDecorr~\cite{DBLP:conf/iclr/ShiLZTB23}, and Gradient-type FedHybrid~\cite{DBLP:journals/tsp/NiuW23} as baselines. We set the Dirichlet parameter $\alpha$  to 0.1,  and the SNR  to 40 dB, selecting the complete graph as the topology in DFL. In Fig.~\ref{fig:method}, FedGMIR outperforms other baselines in most rounds during the training process of CFL and DFL for the two datasets, suggesting that FedGMIR can enhance generalization by implicitly reducing the global mutual information effectively.

\section{Conclusion}
In this paper, we present information-theoretic generalization bounds for FEEL, covering CFL and topology-aware DFL scenarios in the presence of data heterogeneity and noisy channels. Our theoretical analysis reveals the effects of data heterogeneity, communication topology and channel noise on the generalization of FEEL. Besides, we develop a novel method called FedGMIR to enhance the generalization performance of models via implicitly reducing the global mutual information included in the generalization bounds. Numerical results align with our theoretical findings and demonstrate the effectiveness of the proposed FedGMIR method.

    \clearpage
    \appendix
    \subsection{Proof of Theorem~\ref{thm1}}
    
    \begin{lemma}\label{lemmTV}~\cite{DBLP:journals/corr/abs-2306-03824} Together with Assumption~\ref{assumLip}, $   \mathbb{E}\Vert \nabla\mathcal{L}_{\mu_i}(W_t)-G_t\Vert^2 \leq 4L^2D_i^2$, where $D_i=\frac{1}{2}\int_{z \in \mathcal{Z}}|d\mu_i(z)-\frac{1}{N}\sum_{i=1}^Nd\mu_i(z)|$ and $G_t=\frac{1}{N}\sum_{i=1}^N\nabla\mathcal{L}_{\mu_i}(W_t)$.
    \end{lemma}
    
    \begin{proof}
    	For CFL, recall that the considered model $W=W_T$ and we denote $W^{(T)}=(W_1,W_2,...,W_T)$. We can have $\sum_{i=1}^NI(S_i;W)=\sum_{i=1}^NI(S_i;W_T)\leq \sum_{i=1}^NI(S_i;W^{(T)})$. We focus on the term $\sum_{i=1}^NI(S_i;W^{(T)})$ in the following and we have $ \sum_{i=1}^NI(S_i;W^{(T)}) \leq  \sum_{i=1}^N\sum_{t=1}^T I(W_{t+1};Z_t^i|W_{t}).$ This inequality holds based on the updates $\{W_t\}_{t=1}^T$ with a Markov structure and data processing inequality $I(W_{t+1};S_i)\leq I(W_{t+1};Z^i_t) $ since $S_i \leftrightarrow Z^i_t \leftrightarrow W_{t+1}$. We further have $\sum_{i=1}^N\sum_{t=1}^T I(W_{t+1};Z_t^i|W_{t})=I(W_{t+1};(Z_t)_{i=1}^N|W_{t})$ since the independence of data sources.  
    	
    	For any $t$, $ I(W_{t+1};(Z_t)_{i=1}^N|W_{t})=h(W_{t+1}|W_{t})-h(W_{t+1}|(Z_t)_{i=1}^N,W_{t})$.
    	Given that $W_{t}=w_t$ and translation does not affect the differential entropy, we have $ h(W_{t+1}|W_{t}=w_{t})= h(W_{t+1}-w_{t}|W_{t}=w_{t})$.  Based on this result  and the update rule of CFL, we can further have: $       h(W_{t+1}|W_{t}=w_t)=h(-\frac{\eta_t}{N}\sum_{i=1}^N \nabla\ell(w_{t};Z^i_t)+\epsilon_t|W_{t}=w_{t})$.

    	Notice that $\epsilon_t$ is independent of $-\frac{\eta_t}{N}\sum_{i=1}^N \nabla\ell(w_{t};Z^i_t)$, so we can bound the corresponding  expected squared-norm as: $   \mathbb{E}\big\Vert W_{t+1}-w_t\Vert^2
    	\leq  \frac{\eta_t^2}{N^2}\sum_{i=1}^N\mathbb{E}(\xi_i^2+4L^2D_i+\Vert G_t\Vert^2)+d\sigma^2  \leq \frac{\eta_t^2}{N^2}\sum_{i=1}^N(\xi_i^2+L^2(4D_i^2+1))+d\sigma^2$. These inequalities holds since Assumption~\ref{assumLip},~\ref{assumpbounded} and Lemma~\ref{lemmTV}.

    	Among all bounded random variables $X$ with a  bound $B$, the Gaussian distribution $\mathcal{N}(0,\sqrt{\frac{B}{d}I_d})$ has the largest entropy, given by $\frac{d}{2}\log(\frac{2\pi e B}{d})$, so we can  upper bound $h(W_{t+1}|W_t)$ based on the above results.  Given that $h(W_{t+1}|W_{t},(Z_t)_{i=1}^N)=h(\epsilon_t)$ and the derived bound  of $h(W_{t+1}|W_t=w_t)$ holds for all values $w_{t}$, we  can derive the below upper bound of the concerned mutual information similar to the proof in~\cite{DBLP:conf/isit/PensiaJL18}:
    	\begin{equation}
    		\begin{aligned}
    			&I(W_{t+1};(Z_t^i)_{i=1}^N|W_t) \\
    			&\leq \frac{d}{2}\log[\frac{2\pi e(\frac{\eta_t^2}{N^2}\sum_{i=1}^N(\xi_i^2+L^2(4D_i^2+1))+d\sigma^2)}{d}\frac{d}{2\pi e d\sigma^2}]\\
    			&\leq \frac{\eta_t^2\sum_{i=1}^N[\xi_i^2+L^2(4D_i^2+1)]}{2\sigma^2N^2}.
    		\end{aligned}
    	\end{equation}

    	Based on the above results, we can  derive the generalization bound for CFL via Lemma~\ref{lemma:GFL} and complete the proof.

    \end{proof}
    
    \subsection{Proof of Theorem~\ref{thm2}}
    
    \begin{lemma}~\cite{DBLP:conf/colt/Neu21} \label{lemma1}Let $X$ and $Y$ be bounded  random variables taking values in $\mathbb{R}^d$. Let $\epsilon \sim \mathcal{N}(0,\sigma^2I_d)$ be independent of $X$ and $Y$, the relative entropy  between distributions of $X+\epsilon$ and $Y+\epsilon$ is bounded as: $ D_{kl}(P_{X+\epsilon}\Vert P_{Y+\epsilon}) \leq \frac{1}{2\sigma^2}\mathbb{E}\Vert X-Y\Vert^2$.
    \end{lemma}
    
    \begin{lemma}~\cite{sun2021stability}\label{lemma2}
    	Let $\mathbf{P}\in\mathbb{R}^{N \times N}$ be the matrix whose elements are all $1/N$. Given any $k\in \mathbb{Z}^{+}$, the doubly stochastic matrix $\boldsymbol{\Theta}\in \mathbb{R}^{N \times N}$ satisfies the operator norm $\Vert \boldsymbol{\Theta}^k- \mathbf{P} \Vert\leq \lambda^k$, where $\lambda:=\max\{|
    	\lambda_2|,|\lambda_N|\}$.
    \end{lemma}

    \begin{proof}
    	Referring to~\cite{DBLP:conf/colt/Neu21}, we   define a ``ghost" local iterates as $\Tilde{W}_{t+1}^i=\sum_{j=1}^N\theta_{i,j}\Tilde{W}_{t}^j-\eta_t\nabla\ell(\Tilde{W}_{t}^i,\Tilde{Z}_t^i)+\Tilde{\epsilon}_t^i$, where $\Tilde{Z}_t^i$ is sampled from another data set $\Tilde{S}_i$  independently drawn from $\mu_i$. The corresponding average model is $\Tilde{W}_{t+1}=\frac{1}{N}\sum_{i=1}^N\Tilde{W}_{t+1}^i$. These ``ghost" iterates are thus constructed via independently drawn data $\{\Tilde{Z}_i^t\}_{i \in [N]}$ and noise $\{\Tilde{\epsilon}_t^i\}_{i \in [N]}$ at each round. Similar to~\cite{DBLP:conf/colt/Neu21}, conditioned on $W^{(t)}=\Tilde{W}^{(t)}=w^{(t)}$, bounding  $I(W_{t+1};(Z_t^i)_{i=1}^N|W^{(t)})$ can be reduced to bound the relative entropy between the distribution of $W_{t+1}$ given $\{Z_t^i\}_{i=1}^N$ and the  distribution of $\Tilde{W}_{t+1}$ as follows:
    	\begin{equation}
    		\begin{aligned}
    			I(W_{t+1};(Z^i_t)_{i=1}^N|W^{(t)})
    			&= D_{kl}(P_{W_{t+1}|W^{(t)},(Z_t^i)_{i=1}^N}\Vert P_{\Tilde{W}_{t+1}|W^{(t)}}).
    		\end{aligned}
    	\end{equation}

    	For simplicity, we also denote $(W_t^i)_{i =1}^N$ as $\mathbf{W}_t$, $(\Tilde{W}_t^i)_{i =1}^N$ as $\Tilde{\mathbf{W}}_t$, $(-\eta_t\nabla\ell(W_t^i,Z_t^i)+\epsilon_t^i)_{i =1}^N$ as   $\boldsymbol{\psi}_t$, and $(-\eta_t\nabla\ell(\Tilde{W}_t^i,\Tilde{Z}_t^i)+\epsilon_t^i)_{i =1}^N$ as  $\Tilde{\boldsymbol{\psi}}_t$ in the following. Besides, we also denote $\nabla\ell(W_{t}^i,Z_t^i)$ as $\mathbf{g}_t^i$ and $\nabla\ell(\Tilde{W}_{t}^i,\Tilde{Z}_t^i)$ as  $\Tilde{\mathbf{g}}_t^i$ below. Based on Lemma~\ref{lemma1} and the update rule of DFL, we can have:
    	\begin{equation}\label{thm2:derived}
    		\begin{aligned}
    			&D_{kl}(P_{W_{t+1}|W^{t},(Z_t^i)_{i =1}^N}\Vert P_{\Tilde{W}_{t+1}|W^t})\\
    			&\leq \frac{N^2}{2\sigma^2}\mathbb{E}\Vert \frac{1}{N}\sum_{i=1}^N \big[\sum_{j=1}^N\theta_{i,j}(W_{t}^j-\Tilde{W}_{t}^j)-\eta_t(\mathbf{g}_t^i-\Tilde{\mathbf{g}}_t^i)\big]\Vert^2\\
    			&= \frac{N^2}{2\sigma^2}\mathbb{E}\Vert \frac{1}{N}\sum_{i=1}^N \big[\sum_{j=1}^N\theta_{i,j}(W_{t}^j-\Tilde{W}_{t}^j)-W_t+W_t-\eta_t(\mathbf{g}_t^i-\Tilde{\mathbf{g}}_t^i)\big]\Vert^2\\
    			&\leq \mathbb{E}\Big[\frac{1}{2\sigma^2}\Vert (\boldsymbol{\Theta}-\mathbf{P})\mathbf{W}_{t}\Vert^2_F+\frac{1}{2\sigma^2}\Vert (\boldsymbol{\Theta}-\mathbf{P})\Tilde{\mathbf{W}}_{t}\Vert^2_F+\frac{\eta_t^2\mathcal{V}}{\sigma^2}\Big]\\
    			&= \mathbb{E}\Big[\frac{1}{2\sigma^2}\Vert (\boldsymbol{\Theta}-\mathbf{P})\boldsymbol{\Theta}\mathbf{W}_{t-1}+ (\boldsymbol{\Theta}-\mathbf{P})\boldsymbol{\psi}_t\Vert^2_F\\
    			&\quad+\frac{1}{2\sigma^2}\Vert (\boldsymbol{\Theta}-\mathbf{P})\boldsymbol{\Theta}\Tilde{\mathbf{W}}_{t-1}+ (\boldsymbol{\Theta}-\mathbf{P})\Tilde{\boldsymbol{\psi}}_t\Vert^2_F+\frac{\eta_t^2\mathcal{V}}{\sigma^2}\Big]\\
    			&\leq\mathbb{E} \Big[\frac{\lambda^2}{2\sigma^2}\Vert (\boldsymbol{\Theta}-\mathbf{P})\mathbf{W}_{t-1}\Vert^2_F+\frac{\lambda^2}{2\sigma^2}\Vert (\boldsymbol{\Theta}-\mathbf{P})\Tilde{\mathbf{W}}_{t-1}\Vert^2_F\\
    			&\quad+\frac{\lambda^2\eta_{t-1}^2\mathcal{V}}{N^2\sigma^2}+\frac{\lambda^2d}{N^2}+\frac{\eta_t^2\mathcal{V}}{\sigma^2}\Big]\\
    			&\leq\mathbb{E}\Big[\sum_{k=1}^{t}\frac{\lambda^{2k}\eta_{t-k}^2\mathcal{V}}{N^2\sigma^2}+ \sum_{k=1}^{t}\frac{\lambda^{2k}d}{N^2}+\frac{\eta_t^2\mathcal{V}}{\sigma^2}\Big]\\
    			&\leq\mathbb{E}\Big[\frac{\lambda_T\mathcal{V}}{N^2\sigma^2}\sum_{k=1}^{t}\eta_{t-k}^2+ \frac{\lambda_Td}{N^2}+\frac{\eta_t^2\mathcal{V}}{\sigma^2}\Big],
    		\end{aligned}
    	\end{equation}
    	where $\lambda_T=\max_{t\in[T]}\sum_{k=1}^t\lambda^{2k}$, $\mathcal{V}=\sum_{i=1}^N[L^2(4D_i^2+1)+\xi_i^2]$ and $\sigma^2=\sum_{i=1}^N\sigma_i^2$. The second inequality holds since Assumption~\ref{assumLip},~\ref{assumpbounded}, Lemma~\ref{lemmTV}.  The third inequality holds based on the fact $(\boldsymbol{\Theta}-\mathbf{P})\boldsymbol{\Theta}=(\boldsymbol{\Theta}-\mathbf{P})(\boldsymbol{\Theta}-\mathbf{P})$ and Lemma~\ref{lemma2}. The fourth inequality holds since we assume $W_0=\mathbf{0}$.
    	
    	Similarly, according to Lemma~\ref{lemma:GFL}, we can bound the generalization error for DFL using the derived results in Eq.~\eqref{thm2:derived}, completing the proof.
    \end{proof}

\clearpage
\bibliographystyle{ieeetr}

\bibliography{IEEEtran}

\end{document}